\documentclass{article}

\usepackage{arxiv}

\usepackage[utf8]{inputenc} 
\usepackage[T1]{fontenc}    
\usepackage{hyperref}       
\usepackage{url}            
\usepackage{booktabs}       
\usepackage{amsfonts}       
\usepackage{nicefrac}       
\usepackage{microtype}      
\usepackage{lipsum}		

\hypersetup{
	colorlinks=true, 
	citecolor=blue} 

\title{The compound product distribution; a solution to the distributional equation $ X=AX+1 $}

\author{
  Arrigo Coen\thanks{AUTHOR: Arrigo Coen, Email: coen@ciencias.unam.mx} \\
  Departamento de Matem\'aticas, Facultad de Ciencias \\
  Universidad Nacional Aut\'onoma de M\'exico\\
  M\'exico, CDMX, Apartado Postal 20-726, 01000, M\'exico\\
  \texttt{coen@ciencias.unam.mx} 
}

  \usepackage{natbib}

\usepackage{amssymb,amsmath,amsthm,amsfonts,epsfig} 

\usepackage{varioref}
\usepackage{hyperref}
\usepackage{cleveref}
\usepackage{graphicx}
\usepackage{caption}
\usepackage{subcaption}  
\usepackage{tikz}
\usepackage{dsfont}

\usepackage{savesym}
\savesymbol{AND}
\usepackage{algorithm}
\usepackage{algorithmic} 
\usepackage{color}
\usepackage{soul}
\usepackage{cancel}

\usepackage{mathtools} 
\newcommand{\prob}[1]{\mathbb{P}\left[#1\right]}
\newcommand{\esp}[1]{\mathbb{E}\left[#1\right]}

\newcommand{\dist}[1]{\mathrm{#1}} 
\newtheorem{thm}{Theorem}

\newcommand{\mb}{\mathbb}
\newcommand{\eqD}{\stackrel{d}{=}} 
\usepackage{booktabs} 
\usepackage{multirow} 

\begin{document}
\maketitle

\begin{abstract}
	The solution of $ X=AX+1 $ is analyzed for a discrete variable $ A $ with $ \mathbb{P}\left[A=0\right]>0 $. Accordingly, a fast algorithm is presented to calculate the obtained heavy tail density. To exemplify, the compound product distribution is studied in detail for some particular families of distributions.
\end{abstract}


\section{Introduction}
The Kesten's stochastic recurrent equation 
\begin{equation}\label{eq_igual_stoc}
X_t=A_tX_{t-1}+B_t,\qquad t\in\mb{Z},
\end{equation} 
where $ (A_t,B_t) $ are independent of $ X $, has many interesting properties  that makes it appealing for applications, \cite{Kesten1973c}. This framework allows flexibility in describing and estimating the conditional distribution including, for instance, fat tails and stationary behaviors. For example, \cite{DeHaan1989a} analyze the ARCH model  
\[ X_t =Z_t(\beta+\lambda X_{t-1}^2)^{1/2},\qquad t\in\mb{Z}, \] 
as a solution of \eqref{eq_igual_stoc} with $  A_t= \lambda Z_t^2 $, $ B_t = \beta Z_t^2 $. Other uses of this structure range from processes with renewal epochs to model growth-collapse behaviors (see \cite{Boxma2011}) to limit theorems to analyze TCP congestion (see \cite{Dumas} and \cite{Guillemin}). A recent monograph that unites the results of this equation is  \cite{Buraczewski}. Probably the most remarkable property about the solutions to Kesten's equation is that, under general assumptions, the tails of $ X_t $ are of power-law-type even for light-tailed input variables $ (A_t,B_t) $ (see \cite{Kesten1973c}). This implies that these models had many interesting applications on regular variation, weak convergence of probability measures and point processes, as is pointed out in \cite{Buraczewski}. 


The main contribution of the present work is a recursive algorithm to calculate the general density of the solution of \eqref{eq_igual_stoc}, when $B=1  $ and the support of $ A $ is restricted to the non-negative integers with $ \prob{A=0}>0 $. Here and subsequently these restrictions over the density of $ A $ are assumed. In other words, in this manuscript it is analyzed the probability structure of $ X_t $ for the recursive equation
\begin{equation}\label{eq_my_stoc}
X_t=A_tX_{t-1}+1,
\end{equation}
which is equivalent to the analysis of the distributional equality
\begin{equation}\label{eq_my_dist}
X\eqD AX+1.
\end{equation}
Under this framework the fractal structure that reside in \eqref{eq_my_dist} is rich enough to allow many probability structures and, at the same time, does not complicate their use to model the behavior of real heavy tailed data. Moreover, the methodology presented here could be used as a new method to obtain heavy tail distributions. 

The remaining part of this document is organized as follows: In Section \ref{Sec_CP_distribution} it is considered some general properties of the density of $ X $ and also it is studied its probability structure under four different families of densities for $ A $. Section \ref{sec_Algorithm} establishes a computational algorithm to calculate its density. The behavior of the method of moments and maximum likelihood estimators are presented in Section \ref{Sec_estimation}. Furthermore, in this section a fitting exercise with real data is presented. Final points and conclusions are deferred to Section \ref{sec_conclusions}. Additionally, in the Appendix are presented the expressions of skewness and kurtosis, for the four studied distributions.

\section{The compound product distribution}\label{Sec_CP_distribution}

The existence of a solution of \eqref{eq_my_dist} is given in Corollary 2.1.2 and Theorem 2.1.3 of \cite{Buraczewski}. These results allows rewriting $ X $ as
\begin{equation}\label{eq_X_en_A}
X\eqD 1+A_1+A_1A_2+A_1A_2A_3+\ldots ,
\end{equation}
for a sequence $ \left\{A_n\right\}_{n\in\mb{ N}} $ of independent identically distributed  (i.i.d.) random variables. This expression is the reason to say that $ X $ has a compound product (CP) distribution. \cite{Mena2012} also consider the behavior of sums of products to define a probability model and, as in this contribution, they obtain a recursive formula to evaluate the model density. In order to limit the summation in \eqref{eq_X_en_A} it is necessary to put some restrictions on the density of $ A $. In particular, the $\prob{A=0}>0 $ assumption implies that with probability one only a finite number of the terms of the series $ \{A_1\cdots A_n\}_{n\ge1} $ are non-zero. This condition is not particularly restrictive and allows to set many distributional behaviors for $ X $, by changing the distribution of $ A $. The next result implies that a subtle shift over the moments of $ A $ gets amplified in the density of $ X $. 

\begin{thm}
	Let $ X $ be the solution of \eqref{eq_my_dist}. If  $ \esp{A^m}<1 $ for $ m\in\mb{ N} $, then
	\begin{equation}\label{eq_espXm}
	\esp{X^m} =\dfrac{\sum_{i=0}^{m-1}\binom{m}{i}\esp{A^i}\esp{X^i} }{1-\esp{A^m}}. 
	\end{equation}
\end{thm}

\begin{proof}
	The binomial expansion applied to \eqref{eq_my_dist} implies
	\[ \esp{X^m} =\sum_{i=0}^m\binom{m}{i}\esp{A^i}\esp{X^i},\]
	and \eqref{eq_espXm} follows.
\end{proof}

\begin{table}[h!]
	\caption{Conditions for mean and variance to be finite for the CP distribution for different distributions for the variable $ A $.}  
	\label{tab_condiciones_E_y_Var_de_X}
	\centering 
	\begin{tabular}{lcc} 
		\toprule[\heavyrulewidth]\toprule[\heavyrulewidth]
		\textbf{$ A\sim $}  & \textbf{Conditions for $ \esp{X}<\infty $} & \textbf{ Conditions for $ \esp{X^2}<\infty $ } \\ 
		\midrule
		$ \dist{Po}(\lambda) \quad \quad $ & $ \lambda<1 $ &  $ \lambda< \dfrac{\sqrt{5}-1}{2} $\\
		
		$ \dist{Bin}(n,p) $ & always for $ n=1 $ & always for $ n=1 $\\
		& $ p<\dfrac{1}{n} $, $ n\ge2 $ & $ p<\frac{1}{2} \sqrt{\frac{5 n-4}{(n-1)^2 n}}-\frac{1}{2 (n-1)} $, $ n\ge2 $\\[1.5ex]
		
		$ \dist{ NB}(r,p) $ &$ \dfrac{r}{1+r}<p $&$ \dfrac{2}{3}<p $, $ r=1 $\\
		&&$ \frac{2 r^2+r}{2 \left(r^2-1\right)}-\frac{1}{2} \sqrt{\frac{5 r^2+4 r}{\left(r^2-1\right)^2}}<p $, $ r\ge2 $\\[1.5ex]
		
		$ \dist{Geo}(p) $ & $ \dfrac{1}{2}<p $ & $ \dfrac{2}{3}<p $\\
		\bottomrule[\heavyrulewidth] 
	\end{tabular}
\end{table}

\Cref{tab_condiciones_E_y_Var_de_X} presents the conditions for finitude of the first and second moments and \Cref{tab_E_y_Var_de_X} presents the mean and variance, for the CP distribution under the four families of distributions. The important point to note from these two tables is the heavy tail of the distribution of $ X $. For instance, to have a finite variance under $ A\sim\dist{Po}(\lambda) $ we need $ \lambda< (\sqrt{5}-1)/2 \approx 0.62$. These tables exemplify how heavy tailed the CP distribution could be.

\subsection{Examples of families of compound product densities}

This section is focus on four different families of distributions for $ A $: Poisson ($ \dist{Po}(\lambda) $, with mean $ \lambda $), binomial ($ \dist{Bin}(n,p) $, with mean $ np $), negative binomial ($ \dist{NB}(r,p) $, with mean $ pr/(1-p)$) and geometric ($ \dist{Geo}(p) $, with mean $ p/(1-p)$) . Although the geometric distribution is a particular case of the negative binomial family, its analysis expose some properties that could be obscure under the negative binomial framework. An example of the versatility of the negative binomial distribution is presented in \cite{Leisen2019}, to model the stationary time series behavior. On the contrary, an analysis of the Bernoulli distribution it is not worthy since in this case $ X $ is a  geometric random variable. 

\begin{table}[h!] 
	\caption{Mean and variance of the CP distribution for different distributions of $ A $. In each case the obtain formulas are equal to $ \infty $ if the denominator are equal or less to zero. \Cref{tab_condiciones_E_y_Var_de_X} present the detailed conditions for the existence of the first and second moments.}  
	\label{tab_E_y_Var_de_X}
	\centering 
	\begin{tabular}{lcc} 
		\toprule[\heavyrulewidth]\toprule[\heavyrulewidth]
		\textbf{$A\sim$}  & \textbf{$\esp{X}$} & \textbf{$ \dist{Var}(X) $ } \\ 
		\midrule
		$ \dist{Po}(\lambda) \quad \quad $ & $ \dfrac{1}{1-\lambda} $ & $ -\dfrac{\lambda }{(\lambda -1)^2 \left(\lambda ^2+\lambda -1\right)}$  \\[1.5ex]
		$ \dist{Bin}(n,p) $ & $ \dfrac{1}{1-n p} $ & $ \dfrac{n (p-1) p}{(n p-1)^2 \left(n^2 p^2-n (p-1) p-1\right)}$  \\[1.5ex]
		$ \dist{NB}(r,p) $ &$ \dfrac{p}{p r+p-r} $& $ \dfrac{(p-1) p^2 r}{(p r+p-r)^2 \left(p^2 \left(r^2-1\right)-p r (2 r+1)+r (r+1)\right)}$\\[1.5ex]
		$ \dist{Geo}(p) $ & $ \dfrac{p}{2 p-1} $ & $ -\dfrac{(p-1) p^2}{(1-2 p)^2 (3 p-2)} $\\
		\bottomrule[\heavyrulewidth] 
	\end{tabular}
\end{table}

It is possible to study the CP density as the result of a geometric compound process with dependent terms. An equivalent representation of the CP distribution is given as the distribution of the random variable $ X $, given as
\begin{equation}\label{eq_X_cond_N}
X=1+\sum_{i=1}^{ N}A_1'A_2'\ldots A_i', 
\end{equation}
where $  N \sim\dist{geo}(p_0)$ and $ \prob{A_j'=n}=p_n/(1-p_0) $, with $ p_n:=\prob{A=n} $. This expression is obtain by conditioning \eqref{eq_X_en_A} to the random variable $ N=\inf\{i\in \mb{ N}: A_i=0   \} $. One advantage of using this representation is that it gives an straight way to calculate the density of $ X $. By conditioning over $  N $ and applying \eqref{eq_X_cond_N},  it is obtained that $ \prob{X=0}=p_0 $, and for $ n\ge1 $,
\begin{align}
\prob{X=n}&=\sum_{k=0}^\infty \prob{ \left. 1+\sum_{i=1}^N A_1'\ldots A_i'=n \right|  N=k }\prob{ N=k} \nonumber \\ 
&=\sum_{k=0}^{n-1} \prob{ \left. \sum_{i=1}^N A_1'\ldots A_i'=n-1  \right|  N=k }p_0(1-p_0)^{k-1}, \label{eq_probX_N}
\end{align}
where the last equality is given by the fact that 
\[ \prob{ \left. \sum_{i=1}^N A_1'\ldots A_i'=n-1  \right|  N=k }=0,\qquad k=n,n+1,\ldots . \]
The problem with \eqref{eq_probX_N} is that the number of combinations of $ A_1',A_2',\ldots,A_N' $ that fulfills  $ \sum_{i=1}^N A_1'\ldots A_i'=n-1 $ is computational expensive to calculate, since is not available a closed form expression. To solve this issue an recursive algorithm to compute $ \prob{X=n} $ is presented in Section \ref{sec_Algorithm}. \Cref{fig_cambio_Pois,fig_cambio_geom} present the changes of the CP density for different values of the parameters.  \Cref{figsub_pois_l_menor_1} shows the change in the Poisson density for values $ \lambda<1 $ (finite mean case) and \Cref{figsub_pois_l_mayor_1} for values $ \lambda>1 $ (infinite mean case). An equivalent analysis is presented in \Cref{fig_cambio_geom} for the geometric distribution.

\begin{figure}[h]
	\begin{subfigure}{0.5\textwidth}
		\includegraphics[width=\linewidth]{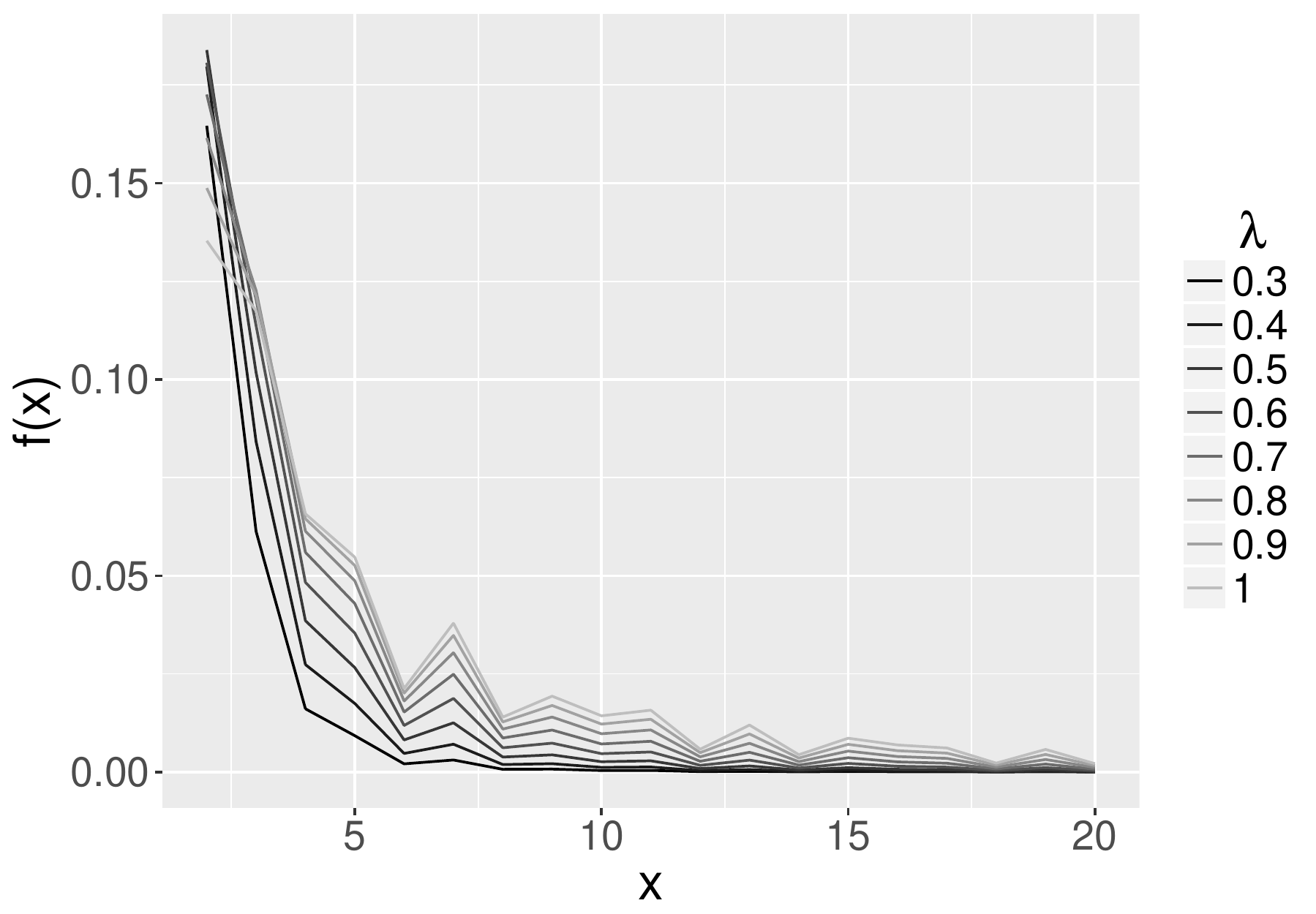} 
		\caption{CP density for $ A\sim\dist{Po} $ with $ \lambda\le1$}
		\label{figsub_pois_l_menor_1}
	\end{subfigure}
	\begin{subfigure}{0.5\textwidth}
		\includegraphics[width=\linewidth]{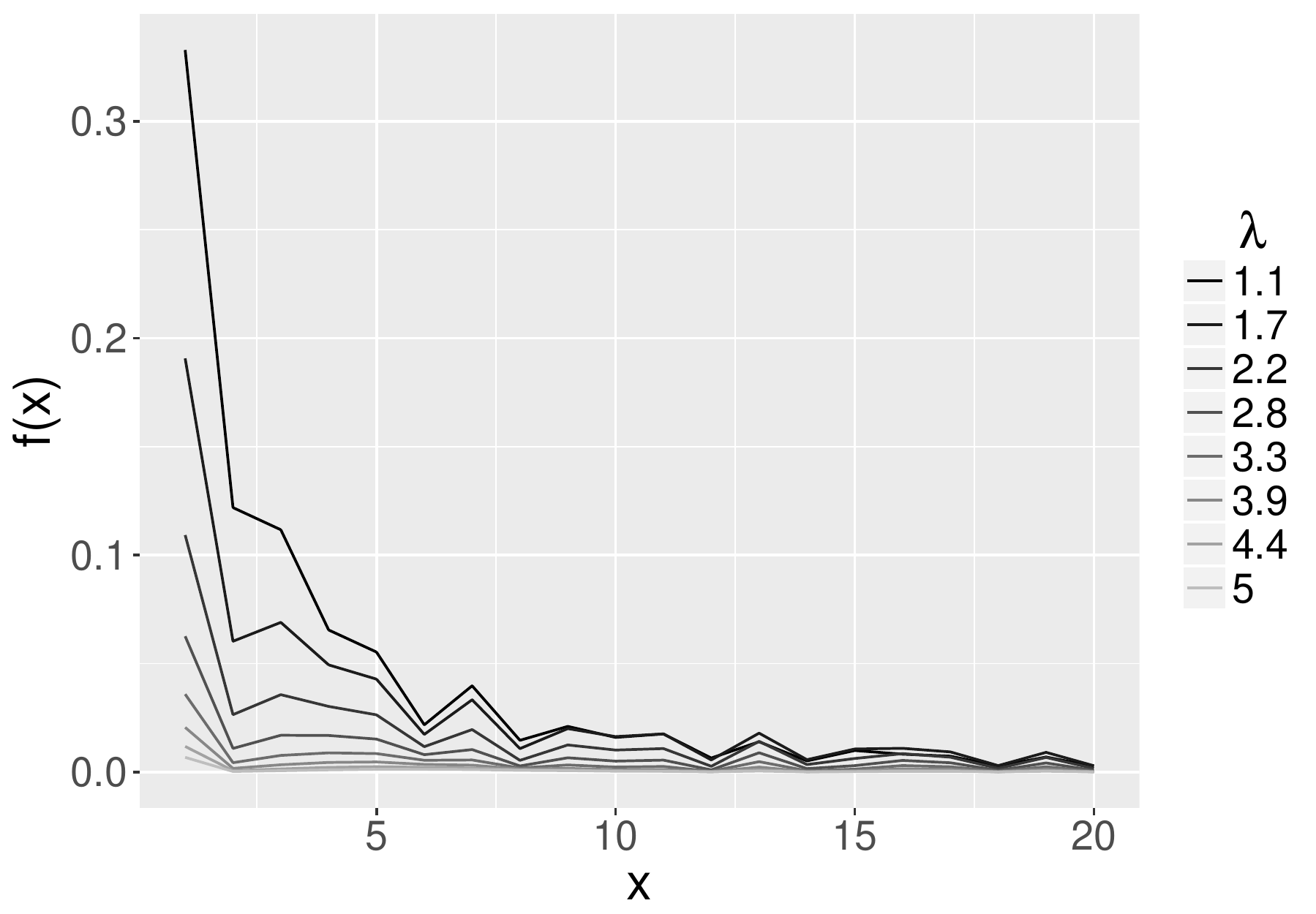}
		\caption{CP density for $ A\sim\dist{Po} $ with $ \lambda>1$}
		\label{figsub_pois_l_mayor_1}
	\end{subfigure}
	\caption{Densities of $ X $ with $ A\sim\dist{Po}(\lambda) $ with different values of $ \lambda $. \Cref{figsub_pois_l_menor_1} present different densities for $ \lambda\le1 $ which implies a finite value for $ \esp{X} $, and \Cref{figsub_pois_l_mayor_1} for $ \lambda>1 $ in which case $ \esp{X}=\infty $.}
	\label{fig_cambio_Pois}
\end{figure}

\begin{figure}[h]
	\begin{subfigure}{0.5\textwidth}
		\includegraphics[width=\linewidth]{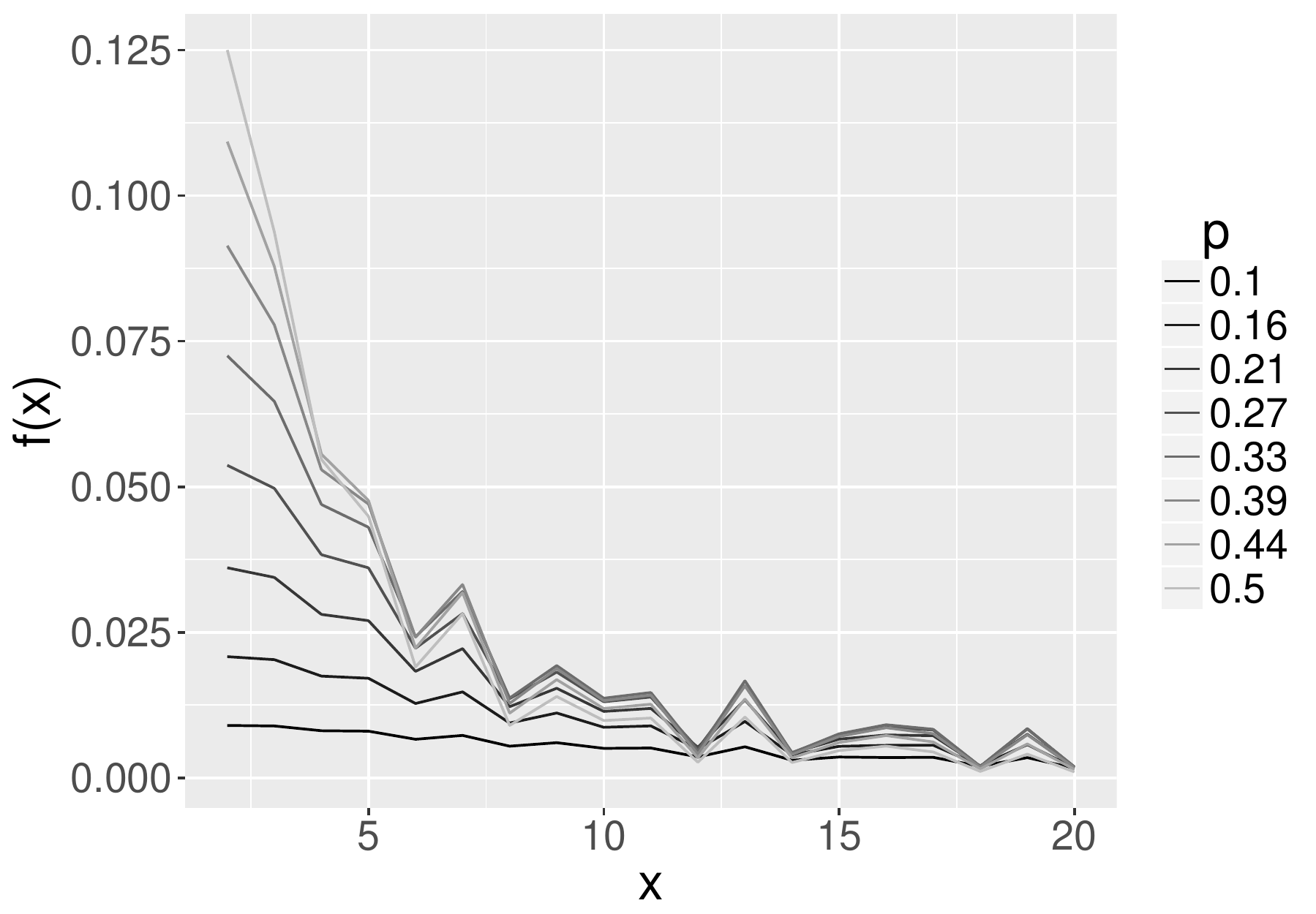} 
		\caption{CP density for $ A\sim\dist{Geo}(p) $ with $ p\le1$}
		\label{figsub_geom_l_menor_1}
	\end{subfigure}
	\begin{subfigure}{0.5\textwidth}
		\includegraphics[width=\linewidth]{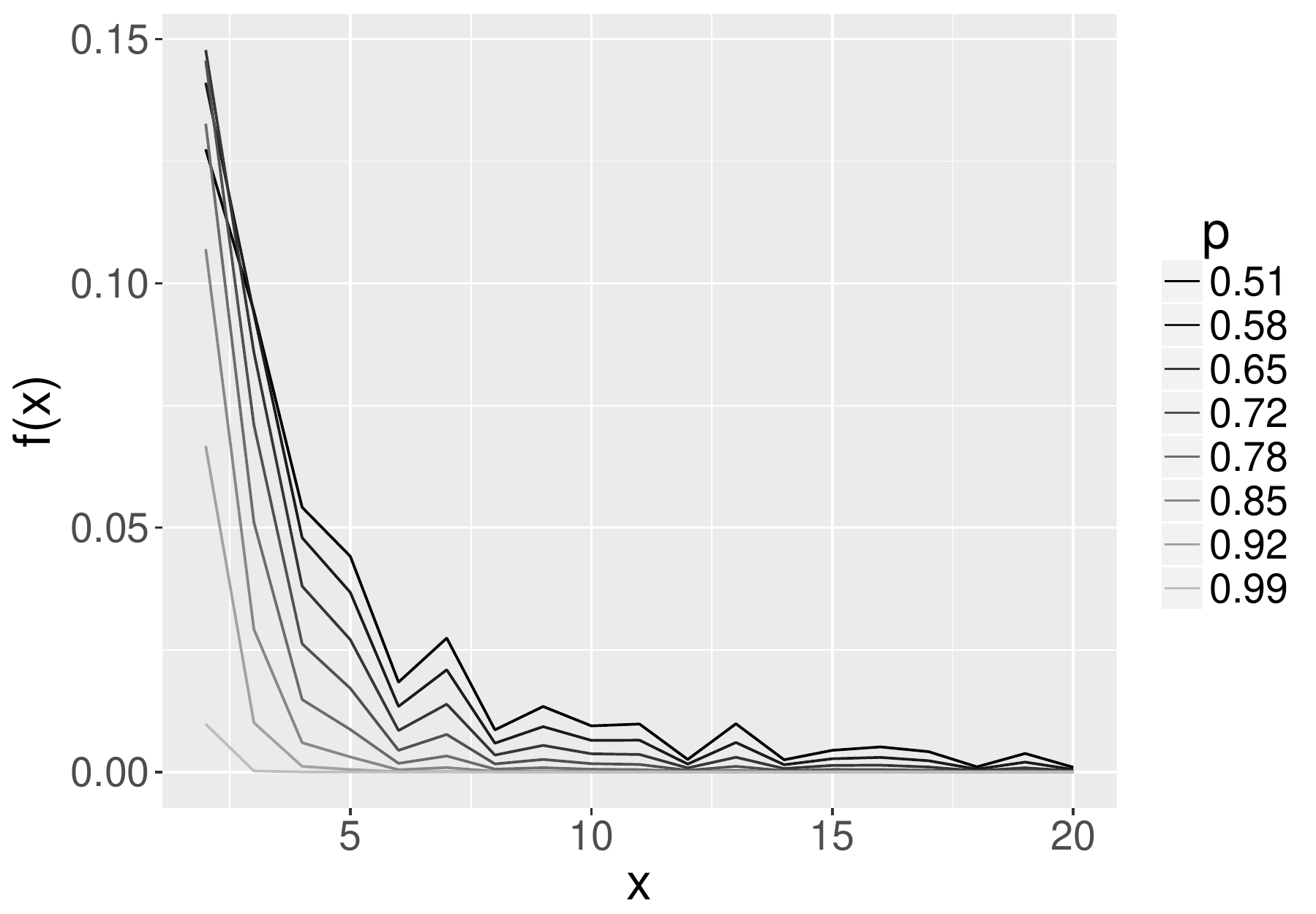}
		\caption{CP density for $ A\sim\dist{Geo}(p) $ with $ p>1$}
		\label{figsub_geom_l_mayor_1}
	\end{subfigure}
	\caption{Densities of $ X $ with $ A\sim\dist{Geo}(p) $ with different values of $ p $. \Cref{figsub_pois_l_menor_1} present different densities for $ p\le1 $ which implies a finite value for $ \esp{X} $, and \Cref{figsub_pois_l_mayor_1} for $ p>1 $ in which case $ \esp{X}=\infty $.}
	\label{fig_cambio_geom}
\end{figure}

\section{An algorithm to calculate the compound product density}\label{sec_Algorithm}

In this section is presented a computational algorithm to calculate $ \prob{X=n}$ for $ n\ge1 $.  The problem to calculate these probabilities centers on the next question: \textit{Which are the different vectors $ (\xi_1,\ldots,\xi_{n-1}) \in\mb{ N}^{n-1}$ that fulfill the next equation?} 
\begin{equation}\label{eq_prods_xi}
n= 1+\xi_1+\xi_1\xi_2+\xi_1\xi_2\xi_3+\ldots+\xi_1\cdots \xi_{n-1}.
\end{equation}

To confront this problem the algorithm here presented focuses around two strategies: prime decomposition and recursion. To explain how the algorithm works, let us first focus on the decomposition of a natural number into two factors. Let us assume that $ n-1 $ has the prime decomposition $ n-1 = z_1^{a_1}\cdots z_r^{a_r} $. This implies that the set $\Delta_n= \{(i,j)\in\mb{ N}^2: i*j=n-1\} $ has cardinality equal to the number of ways to choose how many factors goes to the number $ i $; this cardinality is then equal to $ (a_1+1)( a_2+1)\cdots (a_r+1) $. By conditioning on the values of $ A $, we obtain the recursive equality
\begin{equation}\label{eq_f_X_as_sum}
\prob{X=n} = \sum_{(i,j)\in\Delta_n} p_i  \prob{X=j},\qquad n\ge1,
\end{equation}
since $ X $ and $ A $ are independent. Under these observations is obtained the \Cref{alg_f_X}, which returns all the values $ \prob{X=i} $ for $ i=1,\ldots,n $, for a fixed $ n $. Let us mention some particularities of this algorithm. Using the notation $ \omega(n) $ for the number of different prime factors of $ n $, and $ \Omega(n) $ for the total number of prime factors of $ n $, it is well known (see \cite{hardy1927collected}) that these quantities behave asymptotically as
\begin{equation}\label{eq_omega_y_Omega}
\omega(n)\sim \log\log n \qquad \text{and} \qquad\Omega(n)\sim \log\log n. 
\end{equation}
This two asymptotic behaviors contributes to the speed of the algorithm. Even though \Cref{alg_f_X} has a loop nested on other loop, it is fast since \eqref{eq_omega_y_Omega} give each loop a few terms. These observations are depicted in \Cref{Fig_Time_consumption}. This figure compares the computational time of $\prob{X=n} $ by direct calculation and by \Cref{alg_f_X}. To obtain the direct calculation it is applied \eqref{eq_probX_N}. As a consequence of \eqref{eq_omega_y_Omega}, the time consumption of this algorithm is so suitable that it is almost lineal. For instance, the algorithm takes one second for each thousand values added to $ n $ from $ 1 $ to $ 10^7 $, computed on Intel i7 processor. Also, it is important to mention that with the current computer advances the calculation of the prime decomposition is not an issue of time or memory consumption for this algorithm.

\begin{algorithm}[h!]
	\begin{algorithmic}[1]
		\STATE{ $ p^X\gets (p_0,0,0,\ldots,0) $\hspace{3cm} (a vector of size $ n $) } 
		
		\FOR{ $ i\gets 2 $ \TO $n $} 
		\STATE { $ m \gets $ all factors of $ i-1$ \hspace{2cm} (a two column matrix)}
		\STATE {$ k\gets  $ size of first dimension of $ m $}
		\STATE {$ v\gets (0,0,\ldots,0) $ \hspace{3.4cm}  (a vector of size $ k $) }
		\FOR{ $ j\gets 1 $ \TO $ k $ }\label{for_loop_2}
		\STATE {$v_i\gets p_{m_{i,1}} * p^X_{m_{i,2}}  $ }
		\ENDFOR
		\STATE {$ p^X_i \gets \sum_{i=1}^k v_i $}
		
		\ENDFOR
		\RETURN $ p^X $
	\end{algorithmic}
	\caption{Density of $X $ at points $ \{1,2,\ldots,n\} $ given $ p_0,\ldots,p_n $}
	\label{alg_f_X}
\end{algorithm}

\begin{figure}[!h]
	\centering
	\includegraphics[scale=0.7]{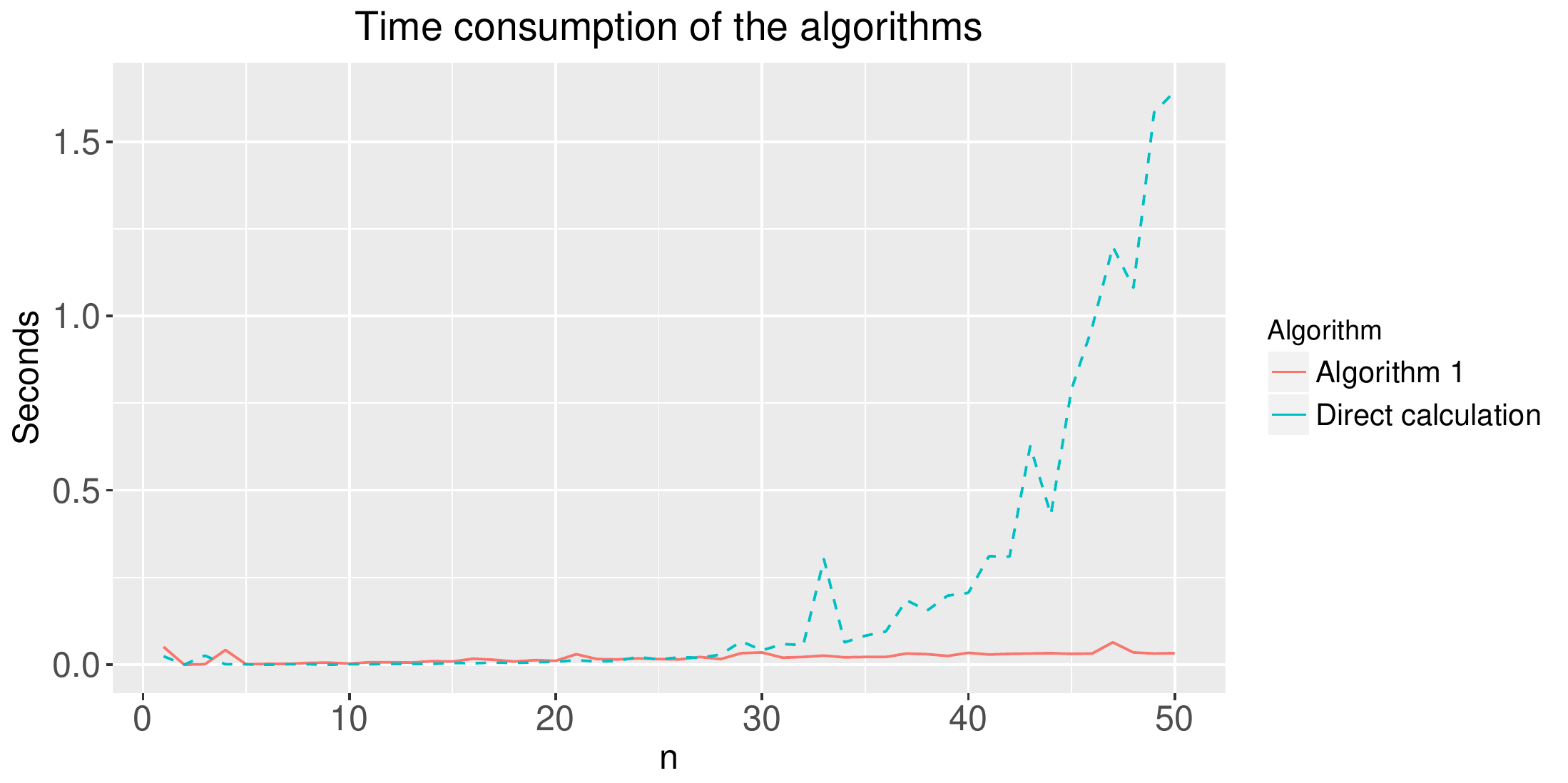}
	\caption{Time comparison to calculate $\prob{X=n} $, by direct calculation and by using \Cref{alg_f_X}. Under \Cref{alg_f_X} the computational time increases almost lineal, in comparison to the exponential time by calculating directly using \eqref{eq_probX_N}.}
	\label{Fig_Time_consumption}
\end{figure}

\section{Estimation of parameters}\label{Sec_estimation}

\subsection{Estimators with the method of moments}

\Cref{tab_method_moments} presents the estimators of method of moments for the CP distribution under the Poisson, binomial, negative binomial and geometric distributions. The existence of these estimators depend on the existence of their respective moments. Consequently, these estimators behave badly when the true parameters are near to conditions that make the correspondent moments infinite (see \Cref{tab_condiciones_E_y_Var_de_X}). 
\begin{table}[h!] 
	\caption{Moment estimators for the CP distribution. For simplicity of notation we denote by $ \mu' = \frac{1}{n}\sum_{i=1}^{n}X_i$ and $ \mu'_2 = \frac{1}{n}\sum_{i=1}^{n}X_i$, the first and second population moments, respectively. }  
	\label{tab_method_moments}
	\centering 
	\begin{tabular}{lll} 
		\toprule[\heavyrulewidth]\toprule[\heavyrulewidth]
		\textbf{$ A\sim $}  & \textbf{Moment matching estimators} \\ 
		\midrule
		$ \dist{Po}(\lambda) \quad \quad $ &$ \hat{\lambda}= 1-\dfrac{1}{\mu'} $ & \\[1.5ex]
		$ \dist{Bin}(n,p) $ & $ \hat{n}=\frac{(\mu' -1)^2 \mu'_2}{(2 \mu' -1) \mu'^2+\left(\mu'^2-3 \mu' +1\right) \mu'_2}$ & $ \hat{p}= \frac{(2 \mu' -1) \mu'^2+\left(\mu'^2-3 \mu' +1\right) \mu'_2}{(\mu' -1) \mu'  \mu'_2} $\\[1.5ex]
		$ \dist{ NB}(r,p) $  & $ \hat{r}= -\frac{(\mu' -1)^2 \mu' _2}{(2 \mu' -1) \mu'^2+\left(\mu'^2-3 \mu' +1\right) \mu'_2}$ & $ \hat{p}= -\frac{(\mu' -1) \mu'  \mu'_2}{(2 \mu' -1) \left(\mu'^2-\mu'_2\right)}$\\[1.5ex]
		$ \dist{ Geo}(p) $ &$ \hat{p}= \dfrac{\mu' }{2 \mu' -1} $&\\
		\bottomrule[\heavyrulewidth] 
	\end{tabular}
\end{table}
The confidence intervals for the moment estimators of the Poisson and geometric distributions are presented in \Cref{figsub_pois,figsub_geom}, respectively. In these figures the black line represents the true value of the estimators and the gray area is its confidence interval from 5\% to 95\%, obtained through simulation of 1000 variables for each parameter value. In both cases the confidence interval is slim since their relative error is small.


\begin{figure}[h!]
	\begin{subfigure}{0.5\textwidth}
		\includegraphics[width=\linewidth]{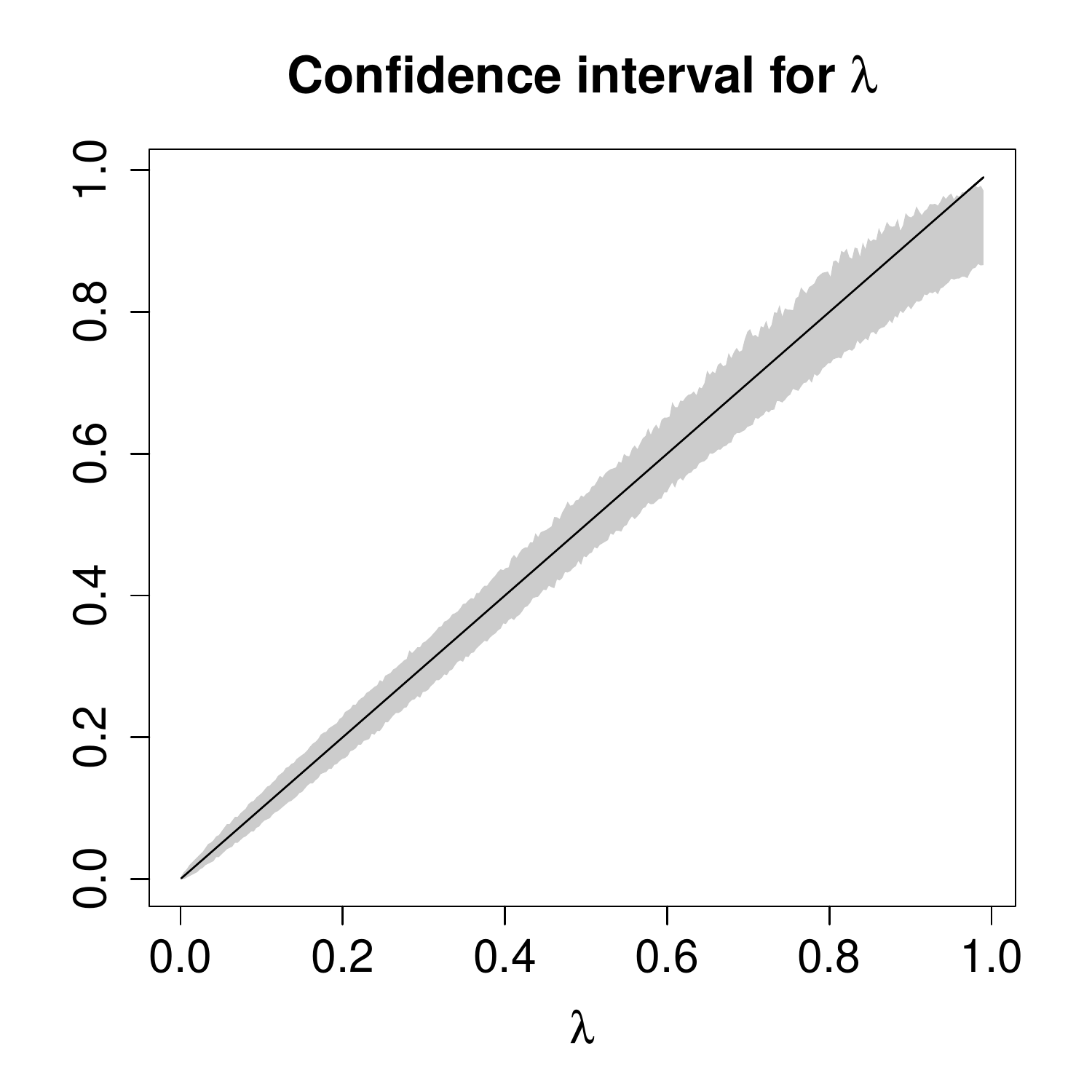} 
		\caption{Poisson distribution}
		\label{figsub_pois}
	\end{subfigure}
	\begin{subfigure}{0.5\textwidth}
		\includegraphics[width=\linewidth]{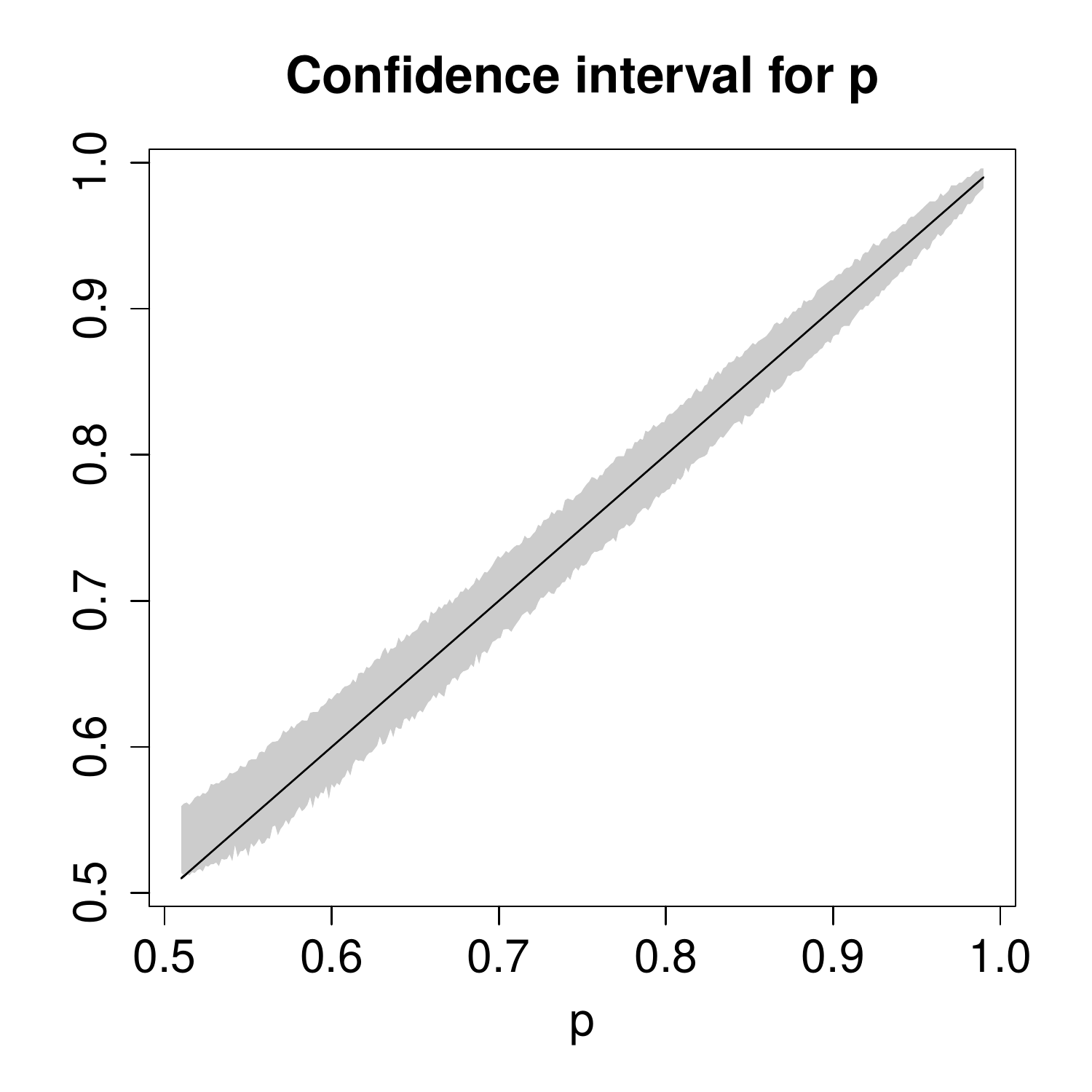}
		\caption{Geometric distribution}
		\label{figsub_geom}
	\end{subfigure}
	\caption{Confidence intervals for the parameter estimation of the method of moments. The black line is the true value and the gray region represents the confidence interval from 5\% to 95\%, obtained through simulation of 1000 variables for each parameter value. \Cref{figsub_pois} presents the confidence interval for the parameter $ \lambda $ with a product distribution Poisson and  \Cref{figsub_geom} the analogous for the parameter $ p $ of a geometric distribution. }
	\label{fig_Interval_Pois_geom}
\end{figure}

In the case of the binomial and negative binomial densities, it is important to notice that extra restrictions must be applied to confine the values of  $ \hat{n} $ and $ \hat{r} $ to the positive integers. For instance, to set these restrictions in the binomial case, one could solve the method of moments equations for the closest integer next  $ \hat{n} $ and then substitute this value in $ \hat{p} $ using the relation  
\[  \hat{p}= \dfrac{1}{\hat{n}\mu'} \]
Similarly, in the negative binomial case
\[ \hat{p} = \hat{r} \frac{\mu  \left((2 \mu -1) \mu ^2+\left(\mu ^2-3 \mu +1\right) \mu _2\right)}{(\mu -1) (2 \mu -1) \left(\mu ^2-\mu _2\right)}. \]
However, these extra restrictions hamper the accuracy of these estimators.

\subsection{Maximum likelihood estimators}

Using \Cref{alg_f_X} one could obtain the MLE estimators, since it is computationally fast enough to compute them. Like with the method of moments, the estimators become inferior as the value of the true parameters tends to the frontiers of a finite mean distribution. In contrast to the method of moments, the MLE method do not have problems to estimate $ \hat{n} $  and $ \hat{r} $, for the binomial and negative binomial distributions. Moreover, the restricting of the searching region to the positive integers speeds-up its computation.

\begin{figure}[h!]
	\begin{subfigure}{0.45\textwidth}
		\includegraphics[width=\linewidth]{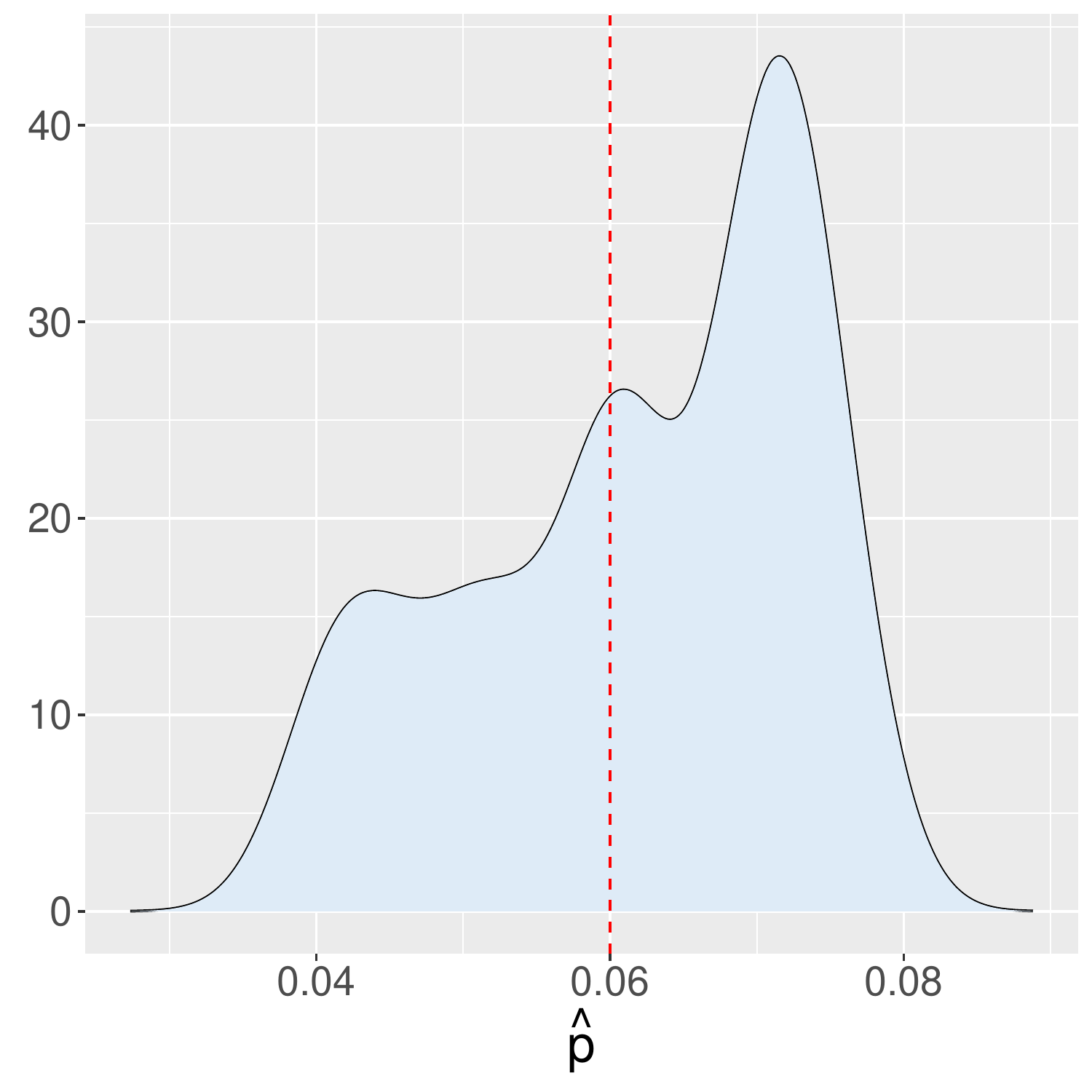} 
		\caption{Binomial distribution}
		\label{figsub_binom}
	\end{subfigure}
	\begin{subfigure}{0.45\textwidth}
		\includegraphics[width=\linewidth]{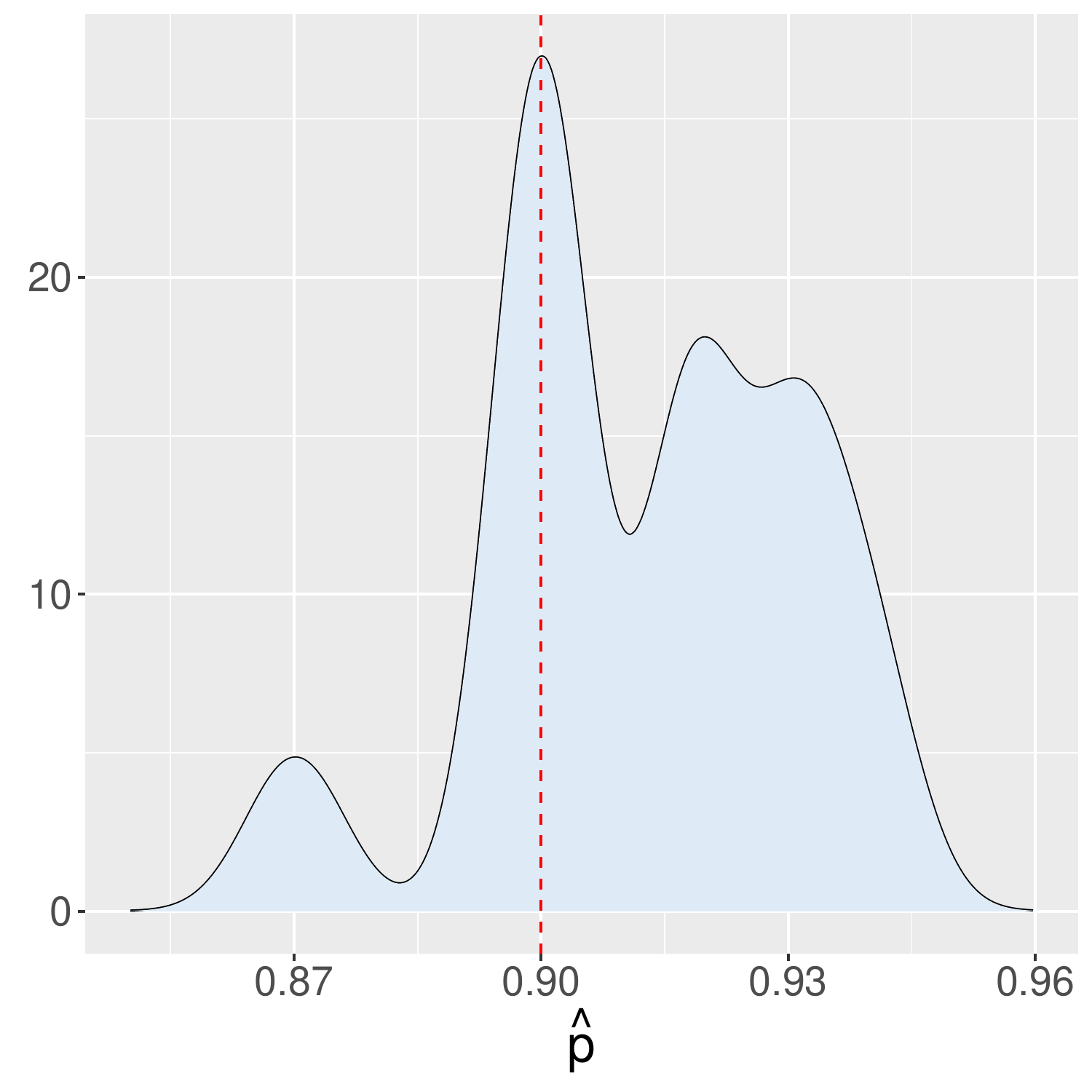}
		\caption{Negative binomial distribution}
		\label{figsub_binneg}
	\end{subfigure}
	\caption{Density of the MLE of $ p $ for the binomial and negative binomial distributions. The curve represents the density of the MLEs; these are the densities of 1000 MLEs of samples of sizes 100. The dotted line is the true value of the $ p $. \Cref{figsub_binom} 
		corresponds to $ \dist{Bin}(6,0.06) $ and \Cref{figsub_binneg} to $ \dist{NB}(4,0.9) $.}
	\label{fig_MLE_bin_nbin}
\end{figure}

\Cref{figsub_binom,figsub_binneg} display a sample density of the MLEs for the binomial and negative binomial distributions, respectively; these densities were obtained through the simulation of 1000 samples of sizes 100 for each point estimator for the fixed true value. In both cases, the densities center around the true value of the distribution. Furthermore, the errors of the estimators $ \hat{n} $  and $ \hat{r} $ were at most one unit apart for the true values 90\% of the time,  for $ n,r\in\{2,3,4,5,6,7,8\} $ and $ p\in(0,1) $.

\subsection{Application to real data}

To present an example of how to fit the CP distribution to real data, here are fitted the four distributional models to the Moby Dick word count data (see \cite{Alsmeyer2013b}). This data is the words occur in the novel Moby Dick, which is well-known to have a heavy-tail behavior. The Poisson and geometric distributions are fitted using the method of moments, and the binomial and negative binomial using the MLE method. To select among the four model it is calculated the Akaike information criterion. The results of the estimation for the Moby Dick data are presented in \Cref{tab_AIC}. Under this criterion the best model is the $ \dist{Bin}(2,.37) $, which has finite mean but infinite variance.

\begin{table}[h!] 
	\caption{Estimated values and Akaike information criterion, for the Moby Dick data. }  
	\label{tab_AIC}
	\centering 
	\begin{tabular}{llll} 
		\toprule[\heavyrulewidth]\toprule[\heavyrulewidth]
		\textbf{$ A\sim  $}  & \textbf{Estimators} && \textbf{AIC} \\ 
		\midrule
		$ \dist{Po}(\lambda) $ & $\hat{\lambda}=0.48$ &  & 174343.18 \\ 
		
		$ \dist{Bin}(n,p) $ & $ \hat{n}=2 $ & $ \hat{p}= 0.37$ & 142566.88 \\ 
		
		$ \dist{NB}(r,p) $ & $\hat{r}= 10  $& $ \hat{p} =0.94 $& 157936.02 \\ 
		
		$ \dist{Geo}(p) $ &$ \hat{p} = 0.67$ &  & 173737.18\\ 
		\bottomrule[\heavyrulewidth] 
	\end{tabular}
\end{table}

\section{Conclusions}\label{sec_conclusions}

In this manuscript is presented a fast algorithm to calculate the general density of the solution of the distributional equation $ X= AX+1$, for a discrete random variable $ A $  with $ \prob{A=0}>0 $. To exhibit the behavior of this distribution, four commonly used parametric families of distributions for $ A $ are particularly studied. The results here presented show that practitioners can use this methodology to generate new  heavy tailed distributions. The evidence of the real data experiment suggests that the estimation method works correctly for the these four families of distributions.

\section{Acknowledgments}

The author gratefully acknowledges the helpful suggestions of Prof. Bego\~na Fern\'andez during the preparation of the manuscript. This research was supported by a DGAPA-UNAM Posdoctoral Scholarship.

\bibliographystyle{plainnat}
\bibliography{ArXiv_CP_dist.bib}
\newpage
\section*{Appendix}
\begin{table}[H]
	\caption{Skewness of the infinite product distribution for different $ A $ densities. In each case the obtain formulas are equal to $ \infty $ if the denominator are equal or less to zero.}  
	\label{tab:example_multirow}
	\centering 
	\begin{tabular}{lc} 
		\toprule[\heavyrulewidth]\toprule[\heavyrulewidth]
		\textbf{$ A\sim $}  & \textbf{$ \dist{Skewness}(X) $}   \\ 
		\midrule
		$ \dist{Po}(\lambda) $ & $\dfrac{(\lambda -1)^3 \left(\lambda ^2+\lambda -1\right)^2 \left(5 \lambda ^2+2 \lambda +1\right)}{\lambda ^2 \left(\lambda ^3+3 \lambda ^2+\lambda -1\right)}$\\[4ex]
		
		$ \dist{Bin}(n,p) $ & $ -\left(\left((n p-1)^3 \left(-n^2 p^2+n (p-1) p+1\right)^2\right.\right. $ \\
		&$ \left.\left.4 (n-1) n p^3+n (6-5 n) p^2-2 (n-1) p-1\right)\right/ $\\
		&$ \left.n^2 (p-1)^2 p^2 \left(n^3 p^3-3 n^2 (p-1) p^2+n \left(2 p^2-3 p+1\right) p-1\right)\right) $ \\[4ex]
		
		$ \dist{NB(r,p)} $ &$ \left((p r+p-r)^3 \left(p^2 \left(r^2-1\right)-p r (2 r+1)+r (r+1)\right)^2\right. $\\
		&$ \left(p r^2 \left(-4 r^2+r-5\right)+r \left(r^3+r+2\right)+p^3 \left(-4 r^4+3 r^3-3 r^2+3 r+1\right)+\right. $\\
		&$ \left.\left.\left.p^4 \left(r^4-r^3+r-1\right)+p^2 \left(6 r^4-3 r^3+7 r^2-6 r+1\right)\right)\right)\right/ $\\
		&$ (p-1)^2 p^3 r^2 \left(p^3 \left(r^2-r+1\right)+p^2 r (2-3 r)+3 p r^2-r (r+1)\right) $\\[4ex]
		
		$ \dist{Geo(p)} $ & $\dfrac{(2-3 p)^2 (2 p-1)^3 \left(6 p^2-13 p+8\right)}{(p-1)^2 p^3 \left(2 p^3-7 p^2+12 p-6\right)} $\\
		\bottomrule[\heavyrulewidth] 
	\end{tabular}
\end{table}
\begin{table}
	\caption{Kurtosis of the infinite product distribution for different $ A $ densities. In each case the obtain formulas are equal to $ \infty $ if the denominator are equal or less to zero.}  
	\label{tab:example_multirow}
	\centering 
	\begin{tabular}{lc} 
		\toprule[\heavyrulewidth]\toprule[\heavyrulewidth]
		\textbf{$ A\sim $}  & \textbf{$ \dist{Kurtosis}[X] $} \\ 
		\midrule
		$ \dist{Po}(\lambda) $ & $ \dfrac{(\lambda -1)^4 \left(\lambda ^2+\lambda -1\right)^3 \left(3 \lambda ^6-25 \lambda ^5-55 \lambda ^4-32 \lambda ^3-47 \lambda ^2-11 \lambda -1\right)}{\lambda ^3 (\lambda +1)^2 \left(\lambda ^2+2 \lambda -1\right) \left(\lambda ^3+5 \lambda ^2+2 \lambda -1\right)} $\\[4ex]
		
		$ \dist{Bin}(n,p) $ &$ \left((n p-1)^4 \left(n^2 p^2-n (p-1) p-1\right)^3\right. $\\
		&$ \left(1+(11 n-6) p+\left(47 n^2-65 n+6\right) p^2+n \left(32 n^2-165 n+138\right) p^3+\right. $\\
		&$ n \left(55 n^3-177 n^2+243 n-120\right) p^4+n \left(25 n^4-184 n^3+339 n^2-216 n+36\right) p^5- $\\
		&$ \left.\left.\left.3 (n-1)^2 n^2 \left(n^2+4 n-12\right) p^7+3 n^2 \left(n^4+10 n^3-62 n^2+93 n-42\right) p^6\right)\right)\right/ $\\
		&$ \left(n^3 (p-1)^3 p^3 \left(n^3 p^3-3 n^2 (p-1) p^2+n \left(2 p^2-3 p+1\right) p-1\right)\right. $\\
		&$ \left.n^4 p^4-6 n^3 (p-1) p^3+n^2 \left(11 p^2-18 p+7\right) p^2+n \left(-6 p^4+12 p^3-7 p^2+p\right)-1\right) $ \\[4ex]
		
		$ \dist{NB(r,p)} $  & $ -\left(\left((p r+p-r)^4 \left(p^2 \left(r^2-1\right)-p r (2 r+1)+r (r+1)\right)^3\right.\right. $\\
		&$\left(r^2 (r+1)^2 \left(3 r^3+5 r^2-3 r+6\right)-2 p r^2 \left(12 r^5+32 r^4+24 r^3+16 r^2+15 r+3\right)+\right.$\\
		&$ p^3 r \left(-168 r^6-106 r^5-140 r^4-29 r^3+76 r^2+35 r+10\right)+ $\\
		&$ p^6 r \left(84 r^6-211 r^5+238 r^4-126 r^3+98 r^2-38 r-24\right)+ $\\
		& $ p^2 r \left(84 r^6+139 r^5+98 r^4+68 r^3-2 r^2-18 r+6\right)+ $\\
		& $ p^4 r \left(210 r^6-85 r^5+210 r^4-103 r^3-9 r^2-66 r-19\right)+ $\\
		& $ p^7 \left(-24 r^7+86 r^6-108 r^5+39 r^4+6 r^3-33 r^2+30 r+1\right)+ $\\
		&$ p^8 \left(3 r^7-14 r^6+20 r^5-4 r^4-16 r^3+20 r^2-7 r-2\right)- $\\
		&$ \left.\left.\left.2 p^5 \left(84 r^7-122 r^6+140 r^5-91 r^4+66 r^3-50 r^2-2 r-1\right)\right)\right)\right/ $\\
		&$ \left((p-1)^3 p^4 r^3 \left(p^3 \left(r^2-r+1\right)+p^2 r (2-3 r)+3 p r^2-r (r+1)\right)\right. $\\
		&$ \left(p^3 r \left(-4 r^2+6 r-3\right)+r \left(r^2+3 r+2\right)-\right. $\\
		&$ \left.\left.\left.p^4 \left(r^3-3 r^2+2 r-1\right)+p^2 \left(6 r^3+r\right)+2 p r \left(2 r^2+3 r+1\right)\right)\right)\right) $\\[4ex]
		
		$ \dist{Geo(p)} $ & $ \dfrac{(1-2 p)^4 (3 p-2)^3 \left(42 p^6-173 p^5+105 p^4+435 p^3-872 p^2+642 p-180\right)}{(p-1)^3 p^4 \left(2 p^3-7 p^2+12 p-6\right) \left(15 p^3-50 p^2+60 p-24\right)} $ \\
		\bottomrule[\heavyrulewidth] 
	\end{tabular}
\end{table}

\end{document}